\RequirePackage{fix-cm}
\documentclass{svjour3}                     % onecolumn (standard format)
\smartqed  % flush right qed marks, e.g. at end of proof
\usepackage{graphicx}
\usepackage{amsmath}

\newtheorem{ex}{Example}

\journalname{ }
\begin{document}

\title{A connected multidimensional maximum bisection problem}

\author{Zoran Maksimovi\'c} 

\authorrunning{Zoran Maksimovi\'c} % if too long for running head

\institute{Z. Maksimovi\'c \at  University of Defence, Military Academy, Pavla Juri\v{s}i\v{c}a \v{S}turma 33, \\11000 Belgrade, Serbia \email{zoran.maksimovic@gmail.com}
 }

\date{Received: date / Accepted: date}
% The correct dates will be entered by the editor

\maketitle

\begin{abstract}
The maximum graph bisection problem is a well known graph partition problem.
The problem has been proven to be NP-hard. 
In the maximum graph bisection problem it is required that the set of vertices 
is divided into two partition with equal number of vertices, and the weight of the edge cut is maximal.

This work introduces a connected multidimensional generalization of the maximum bisection problem. 
In this problem the weights on edges are vectors of positive numbers rather than numbers and partitions should be connected. 
A mixed integer linear programming formulation is proposed with the proof of its correctness. 
The MILP formulation of the problem has polynomial number of variables and constraints. 
\keywords{Graph bisection \and Mixed integer linear programming \and Combinatorial
optimization}
% \PACS{PACS code1 \and PACS code2 \and more}
 \subclass{MSC 90C11 \and MSC 05C70}
\end{abstract}

\section{Introduction}

The maximum bisection problem (MBP) is a well known combinatorial optimization problem.
For a weighted graph $G=(V,E)$ with non-negative weights on the edges and where $|V|$ is an even number, 
the maximum bisection problem consists in finding 
a partition of the set of vertices $V$ in two subsets with the equal cardinality where the sum of weights of the edges 
between the sets is maximal. The maximum bisection can be applied in different fields such 
as VLSI design \cite{slov}, image processing \cite{shij},  compiler optimization \cite{hand}, etc. 

The maximum bisection problem is NP hard as shown in \cite{garr}. 
The complexity of finding optimal and good solutions of the maximum bisection problem 
has given raise to various solution approaches ranging from application algorithms, exact methods to metaheuristics. 

Widely used mathematical formulation with binary variables $x_j$ assigned to each vertex can be presented as 

\begin{flalign*}
\max & \frac{1}{4}\sum_{i,j} w_{ij}(1-x_ix_j)\\
{\rm s.t.\ } &e^T x=0\\
& x^2_j=1,\ \ \ j=1,\ldots ,n\\
\end{flalign*}

where $e\in {R}^n$ is the column vector of all ones, and $^T$ is the transpose operator. It should be noted
that $x_j$ is either $1$ or $-1$ so either $s=\{j|x_j=1\}$ or $s=\{j|x_j=-1\}$. 

This formulation enabled
approximation algorithms based on semidefinite programming, for example in \cite{frieze}, \cite{goemans},
\cite{zwick}, \cite{karish} and \cite{ye}. 
The main goal of these approaches is the performance guarantee so they are not competitive
with other methods for comparison in computational testing.
In paper \cite{hastad} was proved that there is no polynomial approximation algorithm with performance ratio
greater than $\frac{16}{17}$.  

There are several approaches for its exact solving such as 
linear and semidefinite
branch-and-cut methods \cite{armb},
intersection of semidefinite and  polyhedral relaxations \cite{rendl}.

In \cite{armb} is discussed  the minimum graph bisection problem and branch-and-cut approaches
for finding its solution. The problem definition can be described as follows:

Let $G = (V, E)$ be an undirected graph with $V = \{1, \dots , n\}$ and 
edge set $E \subseteq \{\{i, j\} : i, j \in V, i < j\}$. For given vertex weights $f_v \in N \cup \{0\}, v \in V$, and edge costs $w_{i,j} \in R$, $\{i, j\} \in E$, a partition of the vertex set $V$ into two disjoint clusters $S$
and $V\setminus S$ with sizes $f(S) =\sum_{i\in S} f_i \le F$ and $f(V \setminus S) \le F$, 
where $F \in N\cap [\frac12 f(V),f(V)]$,
is called a bisection. Finding a bisection such that the total cost of edges in the
cut $\delta(S) := \{\{i, j\} \in E : i \in S \land j \in V \setminus S\}$ is minimal is the minimum bisection
problem (MB). 

If the function $f$ which represents the weight of nodes is equal to one for all nodes
and $F$ is equal to $\frac12 |V|$ and weights on edges $w_{ij}$ takes negative values
this problem becomes the maximum graph bisection problem. In order to apply brunch-and-cut 
approaches authors in \cite{armb} presented an integer linear programming formulation.

For a selected node $s \in V$ the set of edges can be extended
so that $s$ is adjacent to all other nodes in $V$, where the weights $w$ of new edges is equal to zero.
The extended graph contains a spanning star rooted at $s$. 

Suppose $G$ contains a spanning star rooted at $s$

$$y_{ij} = \left\{
\begin{array}{rl}
	1, &  {\rm if\ } ij {\rm\ is\ in\ the\ cut}\\
	0, & {\rm otherwise,}
\end{array}
\right.
$$

\begin{flalign*}
\min & \sum\limits_{ij} w_{ij}y_{ij}\\
{\rm s.t.\ } &f_s+\sum\limits_{v\ne s} f_v(1-y_{sv})\le F\\
& \sum\limits_{v\ne s} f_v y_{sv}\le F\\
& \sum\limits_{ij\in C\setminus U} y_{ij}+\sum\limits_{ij\in U} (1-y_{ij})\ge 1, {\rm\ \  cycle\ }C\subseteq E,
{\rm odd\ } U\subseteq C\\
& y\in \{0,1\}^E
\end{flalign*}

Semidefinite programming formulation given in \cite{armb} is very similar to the one
already presented in this paper. 
On the basis of large sparse instances coming from VLSI design they
showed the good performance of the semidefinite approach versus the mainstream
linear one.

In the paper \cite{rendl} authors presented a method for finding exact solutions of the Max-Cut problem
based on semidefinite formulation. 
They used semidefinite relaxation combined with triangle inequalities,
which they solve with the bundle method.

Another set of approaches, especially for larger scale instances
are metaheuristics. 
From the wide field of applied metaheuristics
let mention some of them such as: memetic search \cite{wu}, variable
neighbeerhood search \cite{ling}, neural networks \cite{fengmin}, 
deterministing anealing \cite{dang}

Any partition of the node set $V$ in two sets defines a set of edges, that we call a {\it cut}, with ends in different partitions. If a graph has weight on edges, than {\it weight of the cut} 
is defined as the sum of weights of edges in the cut. The problem of finding a partition
of the node set where the weight of the cut is maximal is called a Max-Cut problem. From this
definition it follows that there are no restriction on cardinality of partitions. 
Maximum graph bisection problem is obtained from Max-Cut problem
if it is required that the partitions have equal cardinality. From the definition 
it follows that the Max-Cut problem is a generalization of the maximum graph bisection problem,
and that maximum graph bisection problem can be solved by introducing 
restrictions about cardinality in Max-Cut problem.

In the paper \cite{max1} a multidimensional generalization of maximum bisection problem is proposed, where weights on edges 
instead of numbers are $n$-tuples of positive real numbers. The weight of the cut is the minimum
of sums of  the  coordinates of edge weights. The goal is to find a partition of the set of vertices $V$
in two sets with equal number of vertices and maximal weight of the cut. For $n=1$ we have an ordinary maximum bisection problem. From the fact that maximum bisection problem is $NP$ hard, and that the maximum bisection problem is a special case of the multidimensional maximum bisection problem it follows that multidimensional maximum bisection problem is also $NP$ hard.

The weight of the cut in the multidimensional maximum bisection problem is found by first summing the coordinates  of weight vectors vectors on the edges of the cut. 
After that, the minimum of the sums is determined. Obtained minimum is the weight of the cut. As it can be seen, it is more complex than just summing the weights on the edges of the cut, which is the case in the MBP.

A mixed integer linear
programming formulation with $|V|+|E|$ binary variables and $n+2|E|+1$ constraints is proposed with the proof 
of its correctness. The numerical tests, made on a randomly generated graphs, indicates that the multidimensional
generalization is more difficult to solve than the original problem.

The difficulties of solving this generalization of MBP and inapplicability of solution approaches for classical MBP on generalized problem are discussed in details in \cite{max1}.
Numerical results shown in that paper suggest that this generalization is much harder to solve as can be seen from
comparison of results where dimension $k$ of weight-vectors is equal to 1 and greater then 1.  

In theory and practice it was of interest to consider bisection of graphs where subgraphs are connected. One of the most discussed problem is Maximally Balanced Connected Partition Problem – MBCP, whose formulation can be given as in \cite{matic}:

Let $G=(V, E)$, $V=\{1, 2, \ldots, n\}$  be a connected graph, 
$|E|=m$ and let $w_i$ be weights on vertices. For any subset $V'\subset  V$  the
value $w(V')$ is defined as a sum of weights of all vertices belonging to $V'$ , i.e. $w(V')=\sum_{i\in V'} w_i$. The problem  is to find a partition $(V_1, V_2)$ of $V$ into nonempty disjoint
sets $V_1$ and $V_2$ such that subgraphs of $G$ induced by $V_1$ and $V_2$ are connected and the value $obj(V_1, V_2)=w(V_1)-w(V_2)$ is
minimized. Since subgraphs induced by $V_1$ and $V_2$ are connected they contain spanning trees $T_1(V_1,E'_1)$ and $T_2(V_2,E'_2)$ respectively. Let $p$ and $q$ be arbitrary vertices from $V_1$ and $V_2$ in that order. The spanning trees can be extended in such way that they contain additional vertex 0 with $p$ and $q$ as its only successors, i.e. 
$\overline{E}'_1 = E'_1\cup \{(0,p)\}$,  $\overline{E}'_2 = E'_2\cup \{(0,q)\}$, 
$T=(V,E'_1\cup E'_2)$ and $\overline{T}=(\overline{V},\overline{E}'_1\cup \overline{E}'_2)$ where $\overline{V}=V\cup\{0\}$, $\overline{E}=E\cup \partial E$ and $\partial E=\{(0,i):i\in V\}$.

In order to formulate the MILP model, author in  \cite{matic} introduced variables:

\begin{equation*}
\label{eqxx}
x_{i}=\left\{
\begin{array}{rl}
	1, &  i\in V_1\\
	0, & i \in V_2,	
\end{array}
\right.
\end{equation*}

\begin{equation*}
\label{eqyy}
y_e=\left\{
\begin{array}{rl}
	1, & {\rm\ if\ edge\ } e \in {\overline E}'_1\\
	0, & {\rm\ otherwise,}
\end{array}
\right.
\end{equation*}

\begin{equation*}
\label{eqzz}
z_e=\left\{
\begin{array}{rl}
	1, & {\rm\ if\ edge\ } e \in {\overline E}'_2\\
	0, & {\rm\ otherwise,}
\end{array}
\right.
\end{equation*}

\begin{equation*}
\label{equu}
u_e\in[-n,n], \hfill e\in {\overline E}
\end{equation*}

The MILP formulation for MBCP is:

\begin{equation*}
\label{goalm}
\min -w_{sum}+2\sum\limits_{i=1}^n x_i w_i
\end{equation*}

such that

\begin{equation*}
\label{og1m}
2 \sum\limits_{i=1}^n x_i w_i\ge W,
\end{equation*}

\begin{equation*}
\label{og2m}
y_e \le \frac12 x_{i_e}+\frac12 x_{j_e}, \ \ \ \  (i_e,j_e)=e\in E
\end{equation*}

\begin{equation*}
\label{og3m}
z_e \le 1 - \frac12 x_{i_e} - \frac12 x_{j_e}, \ \ \ \  (i_e,j_e)=e\in E
\end{equation*}

\begin{equation*}
\label{og4m}
y_e \le x_{j_e}, \ \ \ \ \ (i_e,j_e)=e\in \partial E
\end{equation*}

\begin{equation*}
\label{og5m}
z_e \le 1- x_{j_e}, \ \ \ \ \ (i_e,j_e)=e\in \partial E
\end{equation*}

\begin{equation*}
\label{og6m}
u_e \le n\cdot y_e + n\cdot z_e, \ \ \ \ \ e\in \overline{E}
\end{equation*}

\begin{equation*}
\label{og7m}
u_e \ge -n\cdot y_e - n\cdot z_e, \ \ \ \ \ e\in \overline{E}
\end{equation*}

\begin{equation*}
\label{og8m}
\sum\limits_{e:j_e=i} u_e - \sum\limits_{e:i_e=i} u_e = 1,  \ \ \ \ i \in V
\end{equation*}

\begin{equation*}
\label{og9m}
\sum\limits_{e:i_e=0} u_e=n
\end{equation*}

\begin{equation*}
\label{og10m}
\sum\limits_{e\in E} y_e  + \sum\limits_{e\in E} z_e  = n - 2
\end{equation*}

\begin{equation*}
\label{og11m}
\sum\limits_{e\in \partial E} y_e + \sum\limits_{e\in \partial E} z_e = 2
\end{equation*}

\begin{equation*}
\label{og12m}
x_i, y_e, z_e \in \{0,1\}, \ \ \ \ i\in V, e\in \overline{E}
\end{equation*}

\begin{equation*}
\label{og13m}
u_e \in [-n,n], \ \ \ \ e\in \overline{E}
\end{equation*}

where the second and the following constraints are used to ensure connectivity of partitions $V_1$ and $V_2$. 

The first notable theoretical results in analyzing this problem are presented in [1], where the authors proved that
the problem MBCP is NP hard and suggested a simple polynomial time approximation algorithm with a guaranteed
bound 1.072.

MBCP belongs to a wide class of graph partitioning problems and have many applications in different fields of engineering, such as digital signal processing, image processing, managing electric power networks and education. 
Process of controlling and routing in large scale wireless sensor networks is one example, where the network of $N$ clusters is considered, with condition that each cluster corresponds to one cluster head. Division of such network in two balanced subnetworks, with independent optimization of any subnetwork, will simplify the handling process of entire network. 
The network can be represented
as an undirected connected graph, $G=(H, A)$, where $H$ is the set of cluster heads, $H=\{CH_i : i \ldots N\}$ and $A$ is the set of all
undirected links $\{CH_i, CH_j\}$, where $CH_i$ and $CH_j$ are two cluster heads. The objective is to partition G into connected balanced subgraphs and as can be seen this problem is equivalent to  MBCP. In \cite{slama}, the authors adopted the approach proposed in \cite{chleb} and used it to
divide the network of clusters into two smaller, connected sub-networks.

The author of \cite{matic} suggested application in education, where solving MBCP can be useful for finding solutions to practical organizational problems. For example, the
course material can naturally be divided into lessons, where the appropriate difficulty is assigned to each lesson. The connections
between the lessons can be defined by various criteria, like conditionality, analogy or similarity. The idea is to divide
the course material into two disjoint connected sections, so that the sections are of similar difficulty, as much as possible.
Another example would be partitioning the group of students into two smaller groups. The "connectivity" between two students
can be defined in several ways, for example, as "the ability to work together". The objective should be to divide a student
group into two smaller, having in mind that groups should be balanced by student abilities.

\section{Problem definition} 

In this section it will be introduced connected multidimensional max-bisection problem (CMMBP) as
a generalization of multidimensional max-bisection problem. 
As can be seen, the connectivity of subgraphs can be very useful in certain areas of practical and theoretical research. This is the major cause for formulating a new generalization of MBP. 

It can be formulated as follows:    
Let $G=(V,E)$ be an undirected graph, $S\subseteq V$ and $w$ is a function that assigns to 
each edge $e$ an $n$-tuple of positive real numbers $(w_{e1},w_{e2},\ldots,w_{en})$. 
The cut $C(S)$ determined by the set $S$ is defined as \\
$C(S)=\left\{
(i_e,j_e)\in  E|(i_e\in S \land j_e \notin  S) \lor (j_e\in S \land i_e \notin S))
\right\}$. It is obvious that $C(S)=C(V\setminus S)$.
The weight of the cut  
is defined as 
$$w(C(S))=\min\limits_{1\le l\le k} \sum\limits_{(i,j)\in C(S)} w_{ijl}.$$ 
 
The goal of the generalized Max-Bisection problem is to 
find a partition of the set of vertices in two sets with the equal number of vertices where the weight of the cut is maximal, and both partition graphs are connected. This problem will be called Connected Multidimensional Maximum Bisection Problem (CMMBP).

\begin{ex}
\label{ex1}
The connected multidimensional maximum bisection problem can be illustrated by the example given on the  Figure 1, 
which optimal solution is given with the set $S=\{2,3,4\}$.
The set $S$ generates the cut $C(S)=\{(1,2),(1,3),$ $(1,4),$  $(3,5),$ $(3,6),$ $(4,5)\}$
where the sums over coordinates are $(17,16)$ and the weight of the cut is
$16$. The set $S$ and $V\setminus$ are connected with spanning trees $\{(2,3),(3,4)\}$ and $\{(1,6),(5,6)\}$.
\end{ex}
\begin{figure}[h]
	\centering   
		\includegraphics[width=8cm]{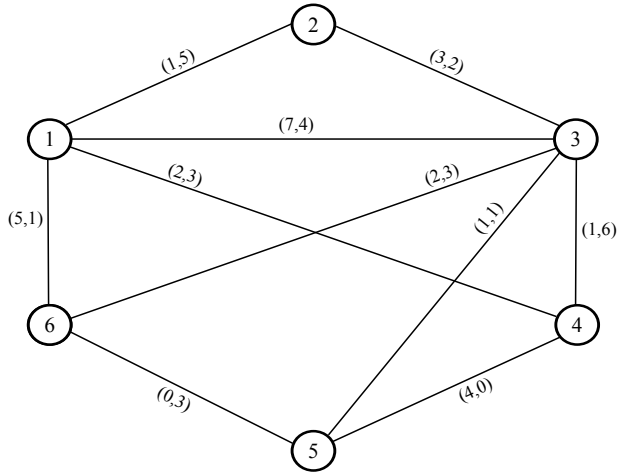}
		\caption{An example of a graph with pairs as weights over the edges}
	\label{fig:example}
\end{figure}

Approaches to solving CMMBP should be considered both from the approaches for solving MBP and those for solving MBCP. As it was presented in \cite{max1} approaches for solving classical MBP are not applicable for various reasons to solving multidimensional maximum bisection problem (MMBP) discussed  in that paper. There it was shown that all three major approaches, approximation algorithms, exact methods and metaheuristics could not be developed without significant modifications into approaches for solving MMBP. As it can be seen CMMBP is special case of MMBP same reasons and considerations presented in \cite{max1} are valid in this case. 

On the other side, graph bisection problems dealing with connectivity are formulated in a form where weights are explicitly linked to vertices and not to edges. Reformulations were weights would be assigned to edges instead of vertices could not adequately address the nature of the problem even if they could be easily executed which is not the case. Furthermore, in formulation weights are vectors instead of numbers, all difficulties mentioned in \cite{max1} concerning this aspect are remaining. This means, that even if there are approaches concerning bisection in connected subgraphs with optimization of objective function dealing with edge weights, such approaches will surely need modifications. Moreover, problems with weight of edges as vectors are more difficult than those where weight are ordinary numbers as can be seen from experimental results presented in \cite{max1}

The solution to the maximum bisection problem can be used in different fields such 
as VLSI design \cite{slov}, image processing \cite{shij},  compiler optimization \cite{hand}, social network analysis etc. 
Connected multidimensional maximum bisection problem appears whenever relation between entities are vectors instead of numbers and the connectivity of subgraphs is essential:

\begin{itemize}
\item 
The proposed problem can be applied in human resource management. One of the most important aspects is compatibility/incompatibility of employees that can be represented by a vector. 
That vector could be character, knowledge, experience, etc where the higher level of incompatibility is represented with greater numbers. 
The employees in that case are represented by vertices and the fact that a certain pair of employees worked together is represented by the edge between those two vertices.  
The problem is to divide the group of employees in two teams with equal size where the greatest part of incompatibility among workers lies between teams.
The connectivity of subgraphs (teams) plays very important factor and that teams are formed by the employees that have worked with each other as much as possible.   

\item The connectivity of electrical components is essential in VLSI design. There are certain aspects that might be considered important such as interference, current used, heat dissipation etc. The proposed problem can be viewed as designation of  electrical components to one of the two boards in such way that both component sets are connected on different boards and maximally differ from each other in the specified number of aspects. For example, the two warmest components should be placed on different board.  
 \end{itemize}

\section{Mixed integer linear programming formulation}

It is well known that is useful to represent problems of the graph theory
as integer programming problems in order to use different well-known exact optimization
techniques. Even more, a mixed integer linear programming model also could be used for developing 
heuristic approaches. Some new development can be seen in \cite{raidl} and \cite{raca}.

In this section it will be introduced a mixed integer linear programming formulation for the connected multidimensional max-bisection problem. The ideas of modeling connectivity of subgraphs follows principles presented in \cite{matic}.

Let $G=(V,E)$ be a graph where $|V|=n$ and let $w_e=(w_{e1},w_{e2},\ldots,w_{ek})$ be a weight vector of an edge $e$. A new vertex $0\notin V$ and new edges $\partial E$ are introduced where 
$\overline V=V\cup\{0\}$, $\overline E = E\cup \partial E$ and $\partial E=\{(0,i)|i\in V\}$.

The goal is to find a partition $(V_1,V_2)$ of the vertex set $V$ where $V_1\cap V_2=\emptyset$, $|V_1|=|V_2|=n/2$ such that $w(C(V_1))$ is maximized and subgraphs induced by $V_1$ and $V_2$ are connected. Let $G_1=(V_1,E_1)$ and  $G_2=(V_2,E_2)$ be such induced graphs. Being connected, they contain spanning trees, 
$T_1=(V_1,E'_1)$ and  $T_2=(V_2,E'_2)$. The edge sets  $E'_1$ and $E'_2$  are extended where 
$\overline{E}'_1={E}'_1\cup\{(0,p)\}$ and $\overline{E}'_1={E}'_1\cup\{(0,q)\}$ for fixed arbitrary vertices $p\in V_1$ and $q\in V_2$. Let graph $T$ and $\overline T$ be defined as $T=(V,E'_1\cup E'_2)$ and 
${\overline T}=(V,{\overline E}'_1\cup {\overline E}'_2)$. 

In order to formulate the MILP model for solving CMMBP, the following variables are introduced:

\begin{equation}
\label{eqx}
x_{i}=\left\{
\begin{array}{rl}
	1, &  i\in V_1\\
	0, & i \in V_2,	
\end{array}
\right.
\end{equation}

\begin{equation}
\label{eqt}
t_e=\left\{
\begin{array}{rl}
	1, & {\rm\ if\ edge\ } e \in C(V_1)\\
	0, & {\rm\ otherwise,}
\end{array}
\right.
\end{equation}

\begin{equation}
\label{eqy}
y_e=\left\{
\begin{array}{rl}
	1, & {\rm\ if\ edge\ } e \in {\overline E}'_1\\
	0, & {\rm\ otherwise,}
\end{array}
\right.
\end{equation}

\begin{equation}
\label{eqz}
z_e=\left\{
\begin{array}{rl}
	1, & {\rm\ if\ edge\ } e \in {\overline E}'_2\\
	0, & {\rm\ otherwise,}
\end{array}
\right.
\end{equation}

\begin{equation}
\label{equ}
u_e\in[-n/2,n/2], \hfill e\in {\overline E}
\end{equation}

Variables $x_i$ determine the vertex set partition, $t_e$ determines edges in the cut needed for calculation of the weight of the cut while  $y_e$ and $z_e$ determines the edges that are in appropriate spanning trees. Variables $u_e$ represents the amount of flow that flows through the certain edge.

The exact solution of the Connected Multidimensional Max-Bisection problem using mixed integer linear programming can be stated as:

\begin{equation}
\label{goal}
\max U
\end{equation}

such that

\begin{equation}
\label{og1}
U \le \sum\limits_{e\in E}w_{el}\cdot t_e,
\ \ \    1\le l\le k
\end{equation}

\begin{equation}
\label{og2}
x_{i_e}+x_{j_e} \ge t_e,
\ \ \    (i_e,j_e)=e\in E
\end{equation}

\begin{equation}
\label{og3}
x_{i_e}+x_{j_e}  + t_e\le 2,
\ \ \    (i_e,j_e)=e\in E
\end{equation}

\begin{equation}
\label{og4}
\sum\limits_{i=1}^{n} x_i = \frac{n}{2},
\end{equation}

\begin{equation}
\label{og5}
y_e \le \frac12 x_{i_e}+\frac12 x_{j_e}, \ \ \ \  (i_e,j_e)=e\in E
\end{equation}

\begin{equation}
\label{og6}
z_e \le 1 - \frac12 x_{i_e} - \frac12 x_{j_e}, \ \ \ \  (i_e,j_e)=e\in E
\end{equation}

\begin{equation}
\label{og7}
y_e \le x_{j_e}, \ \ \ \ \ (i_e,j_e)=e\in \partial E
\end{equation}

\begin{equation}
\label{og8}
z_e \le 1- x_{j_e}, \ \ \ \ \ (i_e,j_e)=e\in \partial E
\end{equation}

\begin{equation}
\label{og9}
u_e \le \frac{n}{2}\cdot y_e + \frac{n}{2}\cdot z_e, \ \ \ \ \ e\in \overline{E}
\end{equation}

\begin{equation}
\label{og10}
u_e \ge -\frac{n}{2}\cdot y_e - \frac{n}{2}\cdot z_e, \ \ \ \ \ e\in \overline{E}
\end{equation}

\begin{equation}
\label{og11}
\sum\limits_{e:j_e=i} u_e - \sum\limits_{e:i_e=i} u_e = 1,  \ \ \ \ i \in V
\end{equation}

\begin{equation}
\label{og12}
\sum\limits_{e:i_e=0} u_e=n
\end{equation}

\begin{equation}
\label{og13}
\sum\limits_{e\in E} y_e  = \frac{n}{2}-1
\end{equation}

\begin{equation}
\label{og14}
\sum\limits_{e\in E} z_e = \frac{n}{2}-1
\end{equation}

\begin{equation}
\label{og15}
\sum\limits_{e\in \partial E} y_e + \sum\limits_{e\in \partial E} z_e = 2
\end{equation}

\begin{equation}
\label{og16}
x_i, t_e, y_e, z_e \in \{0,1\}, \ \ \ \ i\in V, e\in \overline{E}
\end{equation}

\begin{equation}
\label{og17}
u_e \in [-n/2,n/2], \ \ \ \ e\in \overline{E}
\end{equation}

By constraint (\ref{og1}) weight of the cut is determined. Constraints (\ref{og2}) and (\ref{og3}) determines that appropriate edges are in the cut. Constraint (\ref{og4}) ensures that partitions $V_1$ and $V_2$ have equal number of vertices. Constraints (\ref{og5})-(\ref{og14}) ensures that induced subgraphs are connected.  Constraints (\ref{og5}) and (\ref{og6}) determines if the edge is in the spanning trees $T_1$ or $T_2$ or not.
Constraints (\ref{og7}) and (\ref{og8}) ensures that that there is only one edge from additional vertex $0$ to the certain vertex of each spanning tree.  Constraints (\ref{og9}) and (\ref{og10}) ensure that $u_e=0$ for edges $e$ that doesn't belong to the overall spanning tree $\overline T$. The constraint (\ref{og11}) represents network preservation principle. The constraint (\ref{og12}) ensures that there is enough initial flow to reach to all vertices of the graph. The constraints (\ref{og13}) and (\ref{og14}) ensures that there are $n/2-1$ edges in both spanning trees, while the constraint (\ref{og15}) ensures that there exactly two edges emanating from the additional vertex $0$.

It should be noted that constraints (\ref{og5})-(\ref{og8}), (\ref{og11}), (\ref{og12}) and (\ref{og15}) are the same as in formulation for maximally balanced connected partition problem
presented in \cite{matic}. 
Using the idea from Mati\'{c} in \cite{matic} for ensuring  connectivity of partitions the following lemma where appropriate constraints are modified in order to obtain partitions with equal number of vertices.

\begin{lemma}
From constraints (\ref{og5})-(\ref{og15}) it follows that in optimal solution the vertices $p\in V_1$ and $q\in V_2$, such that $u_{(0,p)}=|V_1|$, $u_{(0,q)}=|V_2|$ and $u_{(0,i)}=0$ for $i\ne p,q$ exist.
\end{lemma}
\begin{proof}
Let $e$ be an arbitrary edge from $\overline E$. The edge $e$ may or may not be included in the spanning tree.
In the first case, either $y_e$ or $z_e$ are equal to 1, otherwise $y_e$ and $z_e$ are both equal to $0$. By the constraints (\ref{og5}), values $y$ are bounded by the right hand side of the inequalities by the values $x_{i_e}$ and $x_{j_e}$ (similarly, the constraints (\ref{og6}) bound the variables $z$). For example, if both $i_e$ and $j_e$ belongs to $V_1$, both $x_{i_e}$ and $x_{j_e}$ at the same time are equal to 1, and in that case, constraint (\ref{og5}) allows that the edge $e$ can be included in the spanning tree. Being binary variables, there are  four 
possibilities for  $x_{i_e}$ and $x_{j_e}$ :

\begin{enumerate}
	\item[(i)] $x_{i_e}=1$ and $x_{j_e}=1$: That means that both vertices belong to $V_1$;
	\item[(ii)] $x_{i_e}=1$ and $x_{j_e}=0$: the vertex $i$ belong to $V_1$ and the vertex $j$ belong to $V_2$;
	\item[(iii)] $x_{i_e}=0$ and $x_{j_e}=1$: the vertex $i$ belong to $V_2$ and the vertex $j$ belong to $V_1$;
	\item[(iv)] $x_{i_e}=0$ and $x_{j_e}=0$:  both vertices belong to $V_2$;
\end{enumerate}

Only in the cases (i) and (iv), there exist the possibility that the edge $e$ is included in the spanning tree and in cases (ii) and (iii), constraints (\ref{og5}) and (\ref{og6}) guarantee that $e$ will not be included, and because of (\ref{og9}) and (\ref{og10}), $u_e$ is equal to 0.

From (\ref{og15}), it directly follows that exactly two edges $(0,p)$ and $(0,q)$ exist. These edges belong to $\partial E$ and for them $y_e=1$ or $z_e=1$. It will be shown that $p$ and $q$ have to belong to different subsets $V_1$ and $V_2$. Suppose that, without loss of generality, both $p$ and $q$ belong to $V_1$. From (\ref{og12}) it follows that $u_{(0,p)}+u_{(0,q)}=n$.

Let us sum the constraints from (\ref{og11}), when $i\in V_1$.

\begin{eqnarray*}
|V_1| & = & \sum\limits_{i\in V_1}\left(\sum\limits_{e:j_e=i}u_e-\sum\limits_{e:i_e=i}u_e \right) = \sum\limits_{e:j_e\in V_1}u_e-\sum\limits_{e:i_e\in V_1}u_e = \\
      & = & \sum\limits_{e:i_e\in V_1\land j_e\in V_1}u_e + \sum\limits_{e:i_e=0\land j_e\in V_1}u_e+\sum\limits_{e:i_e\in V_2\land j_e\in V_1}u_e - \\
			& & \left( \sum\limits_{e:i_e\in V_1\land j_e\in V_1}u_e + \sum\limits_{e:i_e\in V_1\land j_e\in V_2}u_e\right) \\
			& = & \sum\limits_{e:i_e=0\land j_e\in V_1}u_e=u_{(0,p)}+u_{(0,q)}=n
\label{eq:}
\end{eqnarray*}

Because the first and the fourth sum are annihilated, while the third and the last sum have all members equal to zero as it is shown above. So, the expression $|V_1|=n$ is obtained which is in contradiction with the fact $|V_1|=\frac{n}{2}$ that follows form the constraint (\ref{og4}). Therefore, $p$ and $q$ must belong to different subsets $V_1$ and $V_2$. 
\end{proof}

\begin{theorem}
A partition of the set of vertices $V'=(V'_1,V'_2)$ is the solution of the CMMBP if and only if constraints (\ref{og1})-(\ref{og17}) and objective function are satisfied for the variables $(x,t,y,z,u)$.
\end{theorem}
\begin{proof}
($\Rightarrow$) 
Based on the bisection $V'$, the variables $(x,t,y,z,u)$ will be constructed and it will be shown that objective function and constraints  (\ref{og1})-(\ref{og17})  are satisfied. Let the variables be defined according to (\ref{eqx})-(\ref{eqz}). It is obvious that (\ref{og16}) is satisfied. 
Variables $u_e, e\in \overline E$ are defined as follows.
For $e\in \partial E$, $u_e$ is defined as

\begin{eqnarray}
u_e=\begin{cases} |V'_1| & e_j=p\\ |V'_2| & e_j= q\\0 & \mbox{otherwise} \end{cases} \label{equovi}
\label{eq3u}
\end{eqnarray}

Since $|V'_1|=\frac{n}{2}$ and $|V'_2|=\frac{n}{2}$, then $u_e\in [-\frac{n}{2},\frac{n}{2}]$  for all $e\in \partial E$.

To determine the values $u_e$ for all other edges $e\in \overline{E}\setminus \partial E$, the fact that in a tree each vertex can be declared as a root is used and all other vertices can be ordered, by some search algorithm. For example, one of the two standard search algorithms can be used for this purpose: depth-first search or breadth-first search. 
Let $E'_1$ and $E'_2$ be the spanning trees for $V'_1$ and $V'_2$ generated by the search algorithm. Using the ordering formed by the search algorithm, in a case when the vertex $i_e$ is a parent of the vertex $j_e$, the value of the flow is defined as the number of vertices in subtree, rooted by the vertex $j_e$. In a case when the vertex $j_e$ is a parent of the vertex $i_e$, the value of the flow $u_{i_e,j_e}$ is defined as the minus number of the vertices in the subtree, rooted by the vertex $i_e$. For all the other edges, belonging neither to $E'_1$ nor $E'_2$, $u_e=0$ is proposed. Since that the number of vertices in each subtree is less or equal to $\frac{n}{2}$, then $u_e\in[-\frac{n}{2},\frac{n}{2}]$ for each edge $e\in \overline{E}\setminus \partial E$. Therefore, (\ref{og17}) holds for each $e\in \overline{E}$.

Let us prove that constraints (\ref{og1})-(\ref{og15}) are satisfied.

Based on the definition of weight of the cut, the constraint ($\ref{og1}$) is true, and based on the goal of CMMBP the  (\ref{goal}) also holds.

If $t_e=0$  than ($\ref{og2}$) and ($\ref{og3}$) are obviously true.
If $t_e=1$  than the corresponding edge $e=(i,j)$ belongs to the cut, and exactly one vertex incident to the edge $e$ must be in the set $V_1$, 
so either $x_i=1$ or $x_j=1$ and therefore constraints ($\ref{og2}$) and ($\ref{og3}$) holds.

The constraint ($\ref{og4}$) is obviously fulfilled as it is required that the vertex set is partitioned into two sets with the equal number of vertices.

To prove inequalities  (\ref{og5}), two cases are considered $y_e=0$ and $y_e=1$.

(i) $y_e=0$. The inequalities (\ref{og5}) are satisfied because $x_{i_e}$,$x_{j_e}\ge0$ by the definition. 

(ii) $y_e=1$ implies that $e\in \overline{E}'_1\Rightarrow x_{i_e}=x_{j_e}=1\Rightarrow y_e\le \frac12 x_{i_e}+\frac12 x_{j_e}$.

To prove inequalities   (\ref{og7}), two cases are considered again:

(i) $y_e=0$. The inequalities (\ref{og7}) holds because $x_{j_e}\ge0$ by the definition. 

(ii) $y_e=1$ implies that $e\in \overline{E}'_1$. Since $e\in \partial E$, $e\in \overline{E}'_1\cap\partial E={(0,p)}$, which implies that $j_e=p\in V_1\Rightarrow x_{j_e}=1$, and therefore inequalities (\ref{og7}) holds. 

The inequalities (\ref{og6}) and (\ref{og8}) are prove analogously as the inequalities (\ref{og5}) and (\ref{og7}).

To prove inequalities (\ref{og9}) and (\ref{og10}) two cases are considered:

(i) $e\in \overline{E}'_1\cup \overline{E}'_2$. Right hand sizes of the inequalities (\ref{og9}) and (\ref{og10}) are equal to $\frac{n}{2}$ and $-\frac{n}{2}$ respectively. Since that $u_e \in [-\frac{n}{2},\frac{n}{2}]$, the inequalities (\ref{og9}) and (\ref{og10}) are satisfied. 

(ii) $e\notin \overline{E}'_1\cup \overline{E}'_2$. Then $u_e$ is equal to 0 by the definition. The inequalities (\ref{og9}) and (\ref{og10}) are satisfied because right hand sizes of (\ref{og9}) are non negative, and right hand sizes of (\ref{og10}) are non positive.

To prove inequality (\ref{og11}), without loss of generality, suppose that $i\in V_1$. Let us consider all the edges from $\overline{E}'_1$, starting or ending by the vertex $i$. The one edge of these comes from from the parent vertex, and all the others are successors. For all the other edges incident with $i$ and not belonging to $\overline{E}'_1$, the values $u_e$ are equal to 0
and they have no influence to (\ref{og11}).

$$
\sum\limits_{e:j_e=i}u_e-\sum\limits_{e:i_e=i}u_e   = \sum\limits_{e:j_e=i\land u_e>0}u_e + \sum\limits_{e:j_e=i\land u_e<0}u_e - 
\sum\limits_{e:i_e=i\land u_e>0}u_e - 
 \sum\limits_{e:i_e=i\land u_e<0}u_e =  $$
$$ 
 =  \sum\limits_{e:j_e=i\land u_e>0}|u_e| + \sum\limits_{e:j_e=i\land u_e<0}-|u_e| - 
\sum\limits_{e:i_e=i\land u_e>0}|u_e| - \sum\limits_{e:i_e=i\land u_e<0}-|u_e| =
$$

$$
 =  \sum\limits_{e:j_e=i\land u_e>0} |u_e| + \sum\limits_{e:i_e=i\land u_e<0} |u_e| - 
\left( \sum\limits_{e:j_e=i\land u_e<0} |u_e| + \sum\limits_{e:i_e=i\land u_e>0} |u_e|\right).
$$

In the first two sums, there is only one edge satisfying the conditions and that edge is coming from the parent of the vertex $i$ to the vertex $i$. For that edge, $|u_e|$ is equal to the number of the vertices in the subtree rooted by $i$. In the last two sums, all the edges coming from the vertex $i$ to all successors in spanning tree $T_1$, participate. The total sum of $|u_e|$ for those edges is equal to the total number of vertices of subtrees rooted by each successor of vertex $i$. Since only vertex $i$ participate in the subtree rooted by itself and not rooted by its successors, the value $\sum_{e:j_e=i\land u_e>0} |u_e| + \sum_{e:i_e=i\land u_e<0} |u_e| - 
\left( \sum_{e:j_e=i\land u_e<0} |u_e| + \sum_{e:i_e=i\land u_e>0} |u_e|\right) =1$

The definition  (\ref{equovi}) of $u_e, e\in \partial E$ is used to prove equality  (\ref{og12}):

$\sum\limits_{e:i_e=0}u_e = u(0,p)+u(0,q)=|v'_1|+|v'_2|=n$.

Since $T_1$ and $T_2$ are the spanning trees of $G_1$ and $G_2$, then $\sum_{e\in E} y_e = |\overline{E}'_1|=|V_1|-1$ and $\sum_{e\in E} z_e = |\overline{E}'_2|=|V_2|-1$. That implies that constraints
(\ref{og13}) and (\ref{og14}) are proved.

For $e\in \partial E$, $y_e=1$ for only one $e$, and that is the case when $e=(0,p)$. For all other edges $e\in \partial E, y_e=0$, which implies $\sum_{e\in\partial E} y_e=1$. Similarly,  $\sum_{e\in\partial E} z_e=1$, implying that constraint (\ref{og15}) is satisfied.

($\Leftarrow$)  
Suppose that optimal solution $(x^*,t^*,y^*,z^*,u^*)$ satisfies the conditions  (\ref{og1})-(\ref{og17}) and the objective function. The bisection $(V_1,V_2)$, that represents the solution of the CMMBP, will be constructed.

Let us define 

$V_1 = \{i\in V | x_i = 1 \}$, $\overline{E}_1 = \{ e\in E | y_e =1 \}$,

$V_2 = \{i\in V | x_i = 0 \}$, $\overline{E}_2 = \{ e\in E | z_e =1 \}$ and 

$C(V_1)=\{e\in E|t_e=1\}$.

The set $C(V_1)$ represents the cut that is generated by the set of vertices $V_1$. 

From the constraint ($\ref{og1}$) it follows that 
$$U\le\min\limits_{1\le l\le k} \sum\limits_{\substack{\{i,j\}\in E\\ i\in S, j\notin S}} w_{ijl},$$
meaning that $U\le w(C(S))$ and it follows from the objective function that $U$ is equal to the greatest weight of the cut.

From the constraint ($\ref{og16}$) $t_e$ is either 0 or 1.

If $t_e=1$ then from the constraints ($\ref{og2}$)  and ($\ref{og3}$) it follows that both vertices of the edge $e$ are not in the same set $V_1$ nor set $V_2$.  

If $t_e=0$ then from the constraints ($\ref{goal}$)-($\ref{og3}$) it follows that both vertices of the edge $e$ must be in the same partition set (either $V_1$ or $V_2$). If vertices are in different partitions, than it can be concluded that the weight of the edge $e$ is not included in the weight of the cut, and therefore $U$ is not maximal which contradicts to the supposition that all constraints are fulfilled. From this it follows that vertices of the edge must be in the same partition.  

From the constraint ($\ref{og4}$) it follows that $|V_1|=n/2=|V_2|$ which means 
that the vertex set is partitioned into two sets with the equal number of vertices.

Constraints (\ref{og5}) - (\ref{og8}) ensures that the sets $\overline{E}_1$ and $\overline{E}_2$ are well defined, i.e. all edges from $\overline{E}_1$ have endpoints from $V_1\cup\{0\}$, and all edges from $\overline{E}_2$ have endpoints from $V_2\cup\{0\}$:

(i) $e\in E'_1$ implies $y_e=1$. From the constraint (\ref{og5}) and the binary nature of the variables $x_{i_e}$ and $x_{j_e}$, it follows that $x_{i_e}=1$ and $x_{j_e}=1$, which implies that $i_e, j_e \in V_1$.

(ii) $e\in \partial E \cap \overline{E}'_1$  also implies $y_e=1$. From the constraints (\ref{og7}) $x_{j_e}=1$ follows, which implies that $j_e\in V_1$.

Similarly, constraints (\ref{og6}) and (\ref{og8}) ensure that all edges from $\overline{E}'_2$ have endpoints from $V_2\cup \{0\}$. Constraints (\ref{og13}) and (\ref{og14}) ensures that the total number of edges included in each spanning tree is exactly $n/2-1$.

From inequalities (\ref{og9}) and (\ref{og10}), it follows that $u_e=0$, for all $e\notin \overline{E}'_1 \cup \overline{E}'_2$.

The connectivity of the graph $(V_1,E'_1)$ will now be proved, and the connectivity of the graph $(V_2,E'_2)$ can be proven analogously.  

Suppose that $S'$ and $S''$ are arbitrary subsets of $V_1$, such that $S'\cup S''=V_1$ and $S'\cap S''=\emptyset$, $S',S''\ne \emptyset$. Then $(\exists e\in E'_1) i_e\in S' \land j_e\in S''$ will be proved.

Let us summarize the constraints given in (\ref{og11}), for all $i\in S'$. We get

$\sum\limits_{i\in S'}\left(\sum\limits_{e:j_e\in S'}u_e - \sum\limits_{e:i_e\in S'}u_e \right)=|S'|$.

If the sum from the left side is disassembled, the following expression is gotten:

$$
\sum\limits_{i\in S'}\left(\sum\limits_{e:j_e\in S'}u_e - \sum\limits_{e:i_e\in S'}u_e \right)  = $$
$$
\sum\limits_{e:j_e\in S'\land i_e\in S'}u_e + \sum\limits_{e:j_e\in S'\land i_e\in S''}u_e + 
\sum\limits_{e:j_e\in S'\land i_e=0}u_e + \sum\limits_{e:j_e\in S'\land i_e\in V_2}u_e  - $$
$$
 -  \left( \sum\limits_{e:i_e\in S'\land j_e\in S'}u_e + \sum\limits_{e:i_e\in S'\land j_e\in S''}u_e
+ \sum\limits_{e:i_e\in S'\land j_e\in V_2}u_e\right).
$$

Let us denote summands in the last equation as A, B, C, D, E, F and G. Then the equation ca be written as 

$\sum\limits_{i\in S'}\left(\sum\limits_{e:j_e\in S'}u_e - \sum\limits_{e:i_e\in S'}u_e \right) = A + B + C+ D - (E + F + G)$.

It is obvious that $A=E$ and $D=G=0$ (between $V_1$ and $V_2$ there are no edges in $E'$). Further, $C=0$ or $C=|V_1|$
depending on whether $p\in S'$ or $p\notin S'$, because the Lemma 1 proposes that exactly one vertex ($p$ in this case) from $V_1$  is connected to the vertex $0$. Thus we have

$\sum\limits_{e:j_e\in S'\land i_e\in S''}u_e \ne \sum\limits_{e:i_e\in S'\land j_e\in S''}u_e$.

The last inequality implies $(\exists e)(((i_e\in S')\land (j_e\in S''))\lor((i_e\in S'')\land (j_e\in S'))) u_e \ne 0$. From $u_e\ne 0 \Rightarrow y_e=1 \Rightarrow e\in E'_1$ and the edge connecting $S'$ and $S''$ is found. Thus, the component $(V_1,E'_1)$ is connected. 

Therefore, constructed partition is connected and the weight of the cut is maximal, ie. the partition represents the solution of CMMBP for the graph $G$.
\end{proof}

In order to illustrate proposed mathematical formulation, 
the following example contains the values of all variables. 

\begin{ex}
\label{ex2}
Let us consider the same graph as in Example \ref{ex1}. Optimal solution value is 16, with set $S=\{2,3,4\}$, where the cut is \\
$C(S)=\{(1,2),(1,3),(1,4),(3,5),(3,6),(4,5)\}$. 
The nonzero values of the variables are as follows: 
$x_1=x_5=x_6=1$, 
$t_{(1,2)}=t_{(1,3)}=t_{(1,4)}=t_{(3,5)}=t_{(3,6)}=t_{(4,5)}=1$,
$y_{(0,1)}=y_{(1,6)}=y_{(5,6)}=1$,
$z_{(0,3)}=z_{(2,3)}=z_{(3,4)}=1$,
$u_{(0,1)}=3, u_{(1,6)}=2, u_{(5,6)}=-1$, 
$u_{(0,3)}=3, u_{(2,3)}=-1, u_{(3,4)}=1$.
\end{ex}

\section{Conclusions}

This paper has taken into consideration a generalization of Max-Bisection problem where
weights on the edges are $n$-tuples and the partition sets induces connected subgraphs. A mixed integer linear programming formulation is introduced with proof of its correctness.

In future work it may be useful to develop a metaheuristic approach for solving CMMBP. 
The second direction could be 
taking  into consideration $n$-tuples as weights in several 
related problems, such as Max-Cut, Max $k$-Cut, Max $k$-Vertex Cover, etc.

\end{document}